\documentclass[format=acmsmall, review=false]{acmart}

\usepackage{acme20} 

\usepackage[utf8]{inputenc}
\usepackage{amsmath}
\usepackage{array}
\usepackage{graphicx}
\usepackage[english]{babel}
\usepackage[utf8]{inputenc}
\usepackage{algorithm}
\usepackage[noend]{algpseudocode}
\usepackage{algorithmicx}

 \usepackage{graphicx}
 \usepackage{float}
 \usepackage{tikz, xcolor}
 \usetikzlibrary{arrows.meta}
\usetikzlibrary {shapes}

\usepackage{xspace}
\usepackage{epsfig}
\usepackage{xcolor}
\usepackage{cleveref}
\usepackage{multirow}
\usepackage{makecell}
\usepackage{colortbl}

\usepackage{enumitem}
\usepackage{wasysym}
\usepackage{framed}
\usepackage{pgfgantt,rotating}
\usepackage{subfloat}
\usepackage{multirow}
\usepackage{blindtext}

\algsetblockdefx[Upon]{Upon}{EndUpon}{}{}[1]{\textbf{upon} #1 \textbf{do}}{}

\newcommand{\factive}{f_a}

\def\reviewing{}  

\ifdefined\reviewing
 \newcommand{\kartik}[1]{{\color{magenta} \footnotesize[Kartik: #1] }}
 \newcommand{\DM}[1]{{\color{blue} \footnotesize[Dahlia: #1] }}
 \newcommand{\oded}[1]{{\color{red} \footnotesize[Oded: #1] }}
 \newcommand{\andy}[1]{{\color{cyan} \footnotesize[Andy: #1] }}
\else
 \newcommand{\kartik}[1]{}
 \newcommand{\DM}[1]{}
 \newcommand{\oded}[1]{}
 \newcommand{\andy}[1]{}
\fi

\title{Lumiere: Making Optimal BFT for Partial Synchrony Practical}

\author{Andrew Lewis-Pye}
\affiliation{%
\institution{ \institution{London School of Economics}}
}

\author{Dahlia Malkhi}
\affiliation{%
  \institution{
    \institution{UC Santa Barbara}}
}

\author{Oded Naor}
\affiliation{%
  \institution{
    \institution{ StarkWare}}
}

\author{Kartik Nayak}
\affiliation{%
  \institution{
    \institution{Duke University}}
}

\date{Jan 2023}

\begin{abstract} 
The \emph{view synchronization} problem lies at the heart of many Byzantine Fault Tolerant (BFT) State Machine Replication (SMR)  protocols in the partial synchrony model, since these protocols are usually based on \emph{views}. Liveness is guaranteed if honest processors spend a sufficiently long time in the same view during periods of synchrony, and if the leader of the view is honest. 
Ensuring that these conditions occur, known as \emph{Byzantine View Synchronization (BVS)}, has turned out to be the performance bottleneck of many BFT SMR protocols.

A recent line of work \cite{L22, opodis22} has shown that, by using an appropriate view synchronization protocol, BFT SMR  protocols can achieve $O(n^2)$ communication complexity in the worst case after GST, thereby finally matching the lower bound established by Dolev and Reischuk in 1985 \cite{dolev1985bounds}. However, these protocols suffer from two major issues, hampering practicality: 
\begin{enumerate} 
\item[(i)] When implemented so as to be \emph{optimistically responsive}, even a single Byzantine processor may infinitely often cause $\Omega(n\Delta)$ latency between consecutive consensus decisions. 
\item[(ii)] Even in the absence of Byzantine action, infinitely many views require honest  processors to send $\Omega(n^2)$ messages.
\end{enumerate} 
Here, we present Lumiere, an optimistically responsive BVS protocol which maintains optimal worst-case communication complexity while simultaneously addressing the two issues above:  for the first time,  Lumiere enables BFT consensus solutions in the partial synchrony setting that have $O(n^2)$ worst-case communication complexity, and that eventually always (i.e., except for a small constant number of ``warmup'' decisions)
have communication complexity and latency which is linear in the number of actual faults in the execution. 
 \end{abstract}

\begin{document}

\maketitle

\section{Introduction} \label{intro} 

State machine replication (SMR) is a central topic in distributed computing and refers to the task of implementing a fault-tolerant service by replicating servers and coordinating client interactions with server replicas~\cite{schneider1990implementing}. Driven partly by high levels of investment in `blockchain' technology, recent years have seen interest in developing SMR protocols that work efficiently at scale~\cite{cohen2021byzantine}. In concrete terms, this means looking to minimize the latency and the communication complexity per consensus decision as a function of the number of processors (participants) $n$.

SMR protocols typically aim to achieve Byzantine fault tolerance, i.e., consensus among processors (server replicas) even if a bounded proportion of the processors are Byzantine and behave arbitrarily/maliciously. The partial synchrony model~\cite{DLS88} is a common networking model on which many of these protocols are based, and this model can be seen as a practical compromise between the synchronous and asynchronous communication models. This model assumes a  point in time called the global stabilisation time (GST) such that any message sent at time $t$ must arrive by time $\max\{\text{GST},t\} + \Delta$. While $\Delta$ is known, the value of GST is unknown to the protocol. \emph{Optimal resiliency} in the partial synchrony communication model means tolerating up to $f$ Byzantine processors among $n$ processors, where $f$ is the largest integer less than $n/3$ \cite{DLS88}.

\vspace{0.2cm} 
\noindent \textbf{The view-based paradigm and view synchronization in partial synchrony}.
Many BFT SMR  protocols~\cite{castro1999practical,yin2019hotstuff,kotla2007zyzzyva,malkhi2023hotstuff} in the partial synchrony model employ a \emph{view-based} paradigm.
The instructions for such protocols are divided into \emph{views}, each view having a designated \emph{leader} to drive progress. 
A consensus decision is guaranteed to be reached during periods of synchrony whenever honest  processors spend a sufficiently long time together in any view with an honest  leader. 
The problem of ensuring that processors synchronize for long enough in the same view is known as the \emph{Byzantine View Synchronization (BVS)} problem.

HotStuff~\cite{yin2019hotstuff} was the first BFT SMR protocol to decouple the core consensus logic from a ``Pacemaker'' module that implements BVS, but left the pacemaker implementation unspecified. 
The core consensus logic in HotStuff requires  quadratic communication complexity in the worst case and linear latency, while each view requires only linear complexity and a constant number of rounds. The task of view synchronization therefore becomes the efficiency bottleneck and the key question becomes: \emph{Can we design a Byzantine view synchronization protocol with optimal communication complexity and  latency?}

The desired communication complexity and latency requirements need some elaboration. Worst-case complexity should be $O(n^2)$ to meet the  Dolev-Reischuk bound~\cite{dolev1985bounds}. However, when there are fewer actual faults $f_a \leq f$, one would hope for the worst-case complexity between every pair of consensus decisions to be a function of $f_a$ such that complexity is $o(n^2)$ when $f_a=o(n)$. Likewise, whereas worst-case latency is $O(n\Delta)$, one would hope for latency $o(n\Delta)$ in the face of a small number $f_a$ of actual faults. Furthermore, given that the core consensus logic can provide \emph{optimistic responsiveness}, such a property would be desirable in BVS also. Roughly, this means that the protocol should function at `network speed' if it turns out that the actual number of faults $f_a\leq f$ is 0: if $f_a=0$, the protocol should be live during periods when message delay is less than the given bound $\Delta$, and latency should be a function of the actual (unknown) message delay $\delta$. This is important because the actual message delay $\delta$ may be much smaller than $\Delta$ when the latter value is conservatively set to ensure liveness under a wide range of network conditions. More formally, we can say that a protocol is optimistically responsive if, subsequent to some finite time after GST,  the worst-case latency between synchronized views with honest leaders (each of which will produce a consensus decision) is  $O(\delta)$ in the case that $f_a=0$. Generalizing this, a protocol is \emph{smoothly optimistically responsive} if the corresponding bound is  $O(\Delta f_a +\delta)$ for any value of $f_a\leq f$.

Recent works have addressed the Byzantine view synchronization problem and achieved some of the above described goals~\cite{Cogsworth21,NK20,L22, opodis22,disc22}. However, all of those works either tolerate only benign failures~\cite{Cogsworth21}, do not obtain optimal worst-case communication and latency~\cite{NK20}, require a large communication complexity and latency even when there are a small number of faults~\cite{L22,opodis22,disc22}, or require stronger assumptions such as clock synchronization~\cite{lewis2023fever}. 
Lumiere is the first Byzantine view synchronization protocol in the partial synchrony model to achieve all of these properties simultaneously. They key result is the following:

\begin{theorem} \label{maintheorem} 
Lumiere is a BVS protocol for the partial synchrony model with the following properties: 
\begin{enumerate} 
\item Worst-case communication complexity $O(n^2)$. 
\item Worst-case latency $O(n\Delta)$. 
\item The protocol is smoothly optimistically responsive. 
\item Eventual worst case communication complexity $O(f_an+n)$. 
\end{enumerate} 

\end{theorem} 

All terms described in Theorem~\ref{maintheorem} are formally defined in Section~\ref{modelsec}. Table~\ref{tab:results-summary} compares the relevant performance measures for state-of-the-art protocols. The comparison is described in more detail in Section~\ref{sec:related-work}. 

\vspace{0.2cm}We emphasize that Lumiere is the first BVS protocol for partial synchrony that satisfies (1) and (2) from the statement of Theorem \ref{maintheorem} while also satisfying \emph{either} (3) or (4) (and, in fact, satisfies both of these properties). The basic approach to describing a protocol that is smoothly optimistically responsive while satisfying properties (1) and (2) is to combine techniques from LP22 \cite{L22} and Fever \cite{lewis2023fever}. LP22 achieves worst-case communication complexity $O(n^2)$ by dividing the instructions into \emph{epochs}, where each epoch consists of $f+1$ views. By performing a heavy ($\Theta(n^2)$ communication complexity) synchronization process only once at the beginning of each epoch, one can avoid the need for synchronization within epochs, thereby matching the  Dolev-Reischuk bound~\cite{dolev1985bounds}. LP22 also achieves optimistic responsiveness, but suffers from the issue that even a \emph{single} faulty leader is then able to achieve $\Theta(n\Delta)$ delays. Fever, meanwhile, is a protocol that is optimal for all measures considered in this paper, but which requires synchronized clocks (and so operates in a non-standard model). We show that, through the use of epochs (as in LP22) and by employing  Fever \emph{within} epochs, one can achieve a protocol for the partial synchrony model that simultaneously satisfies properties (1)--(3). 

The most technically complex task is then to modify the protocol so as to also ensure that the eventual worst case communication complexity is $O(f_an+n)$ (property (4)). To do so requires establishing conditions that allow processors to stop carrying out heavy epoch changes once sufficiently synchronized. This involves a number of substantial technical complexities that are discussed in depth in Section \ref{SS}.
 
\begin{table}[t]
\centering 
\setlength{\tabcolsep}{6pt}
\renewcommand{\arraystretch}{1.5}
\resizebox{.7\textwidth}{!}{%
\begin{tabular}{l|c|c|c|>{\columncolor[gray]{0.9}}c|}
\toprule
\textbf{Protocol}                                                                   & 
\makecell{\textbf{Cogsworth}\\ \textbf{NK20}} &
\makecell{ \textbf{LP22}} &
\textbf{Fever}                                              &
\makecell{\textbf{Lumiere}\\ \textbf{(this work)}}                                \\ 
\toprule \toprule
\textbf{Model} &  
\makecell{Partial \\ Synchrony}                                                 &
\makecell{Partial \\ Synchrony}                                                 &
\makecell{Bounded \\ Clocks}                                                   &
\makecell[c]{Partial \\ Synchrony}                                               \\ 
\midrule
\makecell[l]{\textbf{Worst-case}\\ \textbf{Communication}}         &
$O(n^3)$                                                           &
$O(n^2)$                                                           &
$O(n^2)$                                                           &
$O(n^2)$                                                           \\ 
\midrule
\makecell[l]{\textbf{Eventual Worst-case} \\ \textbf{Communication}}      & 
$O(n + n \factive^2)$                                                           &
$O(n^2)$                                                           &
$O(n \factive+n)$                                                             &
$O(n \factive+n)$                                                      \\
\midrule
\makecell[l]{\textbf{Worst-case}\\ \textbf{Latency}}              & 
$O(n^2 \Delta)$                                                          &
$O(n\Delta)$                                                            &
$O(f_a\Delta +\delta)$                                                            &
$O(n\Delta)$                                                            \\ 
\midrule
\makecell[l]{\textbf{Eventual Worst-case}\\ \textbf{Latency}}          &
$O(\factive^2\Delta +\delta)$                                                 &
$O(n\Delta)$                                                          &
$O(\factive \Delta +\delta)$                                                   &
$O(\factive\Delta +\delta)$                                                   \\ 
\bottomrule
\end{tabular}%
}
\vspace{0.2cm}
\caption{\textbf{Summary of the results for state-of-the-art optimistically responsive protocols.}  }
\label{tab:results-summary}
\end{table}

\section{The Setup}  \label{modelsec} 

For simplicity, we assume $n=3f+1$ and consider a set $\Pi = \{p_1, \ldots, p_n \}$ of $n$ processors. At most $f$ processors may become corrupted by the adversary during the course of the execution, and may then display  \emph{Byzantine} (arbitrary) behaviour. We let $f_a$ denote the actual number of processors that become corrupted. Processors that never become corrupted by the adversary are referred to as \emph{honest}.

\vspace{0.2cm} 
\noindent \textbf{Cryptographic assumptions}. Our cryptographic assumptions are standard for papers on this topic. Processors communicate by point-to-point authenticated channels. We use a cryptographic signature scheme, a public key infrastructure (PKI) to validate signatures, and a threshold signature scheme \cite{boneh2001short,shoup2000practical}.  The threshold signature scheme is used to create a compact signature of $m$-of-$n$ processors, as in other consensus and view synchronisation protocols \cite{yin2019hotstuff}. In this paper, either $m=f+1$ or $m=2f+1$.  The size of a threshold signature is $O(\kappa)$, where $\kappa$ is a security parameter, and does not depend on $m$ or $n$.
 We assume a computationally bounded adversary. Following a common standard in distributed computing and for simplicity of presentation (to avoid the analysis of negligible error probabilities), we assume these cryptographic schemes are perfect, i.e.\ we restrict attention to executions in which the adversary is unable to break these cryptographic schemes.  
 
  \vspace{0.2cm} 
\noindent \textbf{The partial synchrony model}. As noted above, processors communicate using point-to-point authenticated channels. We consider the standard partial synchrony model, whereby a message sent at time $t$ must arrive by time $\max\{\text{GST},t\} + \Delta$. While $\Delta$ is known, the value of GST is unknown to the protocol. The adversary chooses GST and also message delivery times, subject to the constraints already defined. We let $\delta$ denote the actual (unknown) upper bound on message delay after GST. 
Each processor $p$  also maintains a local clock value $\mathtt{lc}(p)$. We assume that each processor $p$ may join the protocol with  $\mathtt{lc}(p)=0$ at any arbitrary time prior to GST, and that processors may experience arbitrary clock drift prior to GST. For simplicity we assume that, for honest $p$ after GST,  $\mathtt{lc}(p)$ advances in real time, except when $p$ pauses $\mathtt{lc}(p)$ or bumps it forward (according to the protocol instructions).  However, our analysis is easily modified to deal with a scenario where local clocks have bounded drift during any interval after GST in which they are not paused or bumped forward. When we wish to make the dependence on $t$ explicit, we write  $\mathtt{lc}(p,t)$ to denote the value $\mathtt{lc}(p)$ at time $t$.

 \vspace{0.2cm}
\noindent   \textbf{The underlying protocol}. We suppose view synchronisation is required for some underlying protocol (such as Hotstuff) with the following properties: 
\begin{itemize} 
\item \textbf{Views}. Instructions are divided into views. Each view $v$ has a designated \emph{leader}, denoted $\mathtt{lead}(v)$.

\item \textbf{Quorum certificates}. The successful completion of a view $v$ is marked by all processors receiving a \emph{Quorum Certificate} (QC) for view $v$. The QC is a threshold signature of length $O(\kappa)$ (for the security parameter $\kappa$ that determines the length of signatures and hash values) combining $2f+1$ signatures from different processors testifying that they have completed the instructions for the view. 
In a chained implementation of Hotstuff, for example, the leader will propose a block, processors will send votes for the block to the leader, who will then combine those votes into a QC and send this to all processors. Alternatively, one could consider a (non-chained) implementation of Hotstuff, in which the relevant QC corresponds to a successful third round of voting. 
\item \textbf{Sufficient time for view completion}. We suppose:
\begin{enumerate} 
\item[$(\diamond_1)$]  There exists some known $x\geq 2$ such that if $\mathtt{lead}(v)$ is honest, if (the global time) $t\geq \text{\text{GST}}$, and if at least $2f+1$  honest  processors are in view $v$ from time $t$ until either they receive a QC for view $v$ or until $t+x\delta$, then  all honest processors will receive a QC for view $v$ by time  $t+x\delta$, so long as all messages sent by honest processors while in view $v$ are received within time $\delta \leq \Delta$. 
\item[$(\diamond_2)$] No view $v$ produces a QC  unless there is some non-zero interval of time during which at least $2f+1$ processors all act as if honest and in view $v$. 
\end{enumerate} 
\end{itemize} 

\vspace{0.1cm} 

\noindent \textbf{The view synchronisation task}. For $x$ as above, we must ensure: 
\begin{enumerate} 
\item If an honest processor is in view $v$ at time $t$ and in view $v'$ at $t'\geq t$, then $v'\geq v$. 
\item There exists some honest $\mathtt{lead}(v)$ and $t\geq$ GST such that each honest processor is in view $v$ from time $t$ until either it receives a QC for view $v$ or until $t+x\Delta$.
\end{enumerate}
Condition (1) above is required by standard view-based SMR protocols to ensure consistency. Since GST is unknown to the protocol, condition (2)  suffices to ensure the successful completion of infinitely many views with honest leaders. By a BVS protocol, we mean a protocol which determines when processors enters views and which satisfies conditions (1) and (2) above.

\vspace{0.2cm} 
\noindent \textbf{Complexity measures}. 
All messages sent by honest processors will be of length $O(\kappa)$, where $\kappa$ is the security parameter determining the length of signatures and hash values. We make the following definitions.
 Given $T\geq \text{GST}$, let $t_T^{\ast}$ be the least time $>T$ at which the underlying protocol has some honest $\mathtt{lead}(v)$ produce a QC for view $v$ (if there exists no such time, set $t_T^*:=\infty$). Then: 
 
 \begin{itemize} 
 \item The \emph{worst-case communication complexity after} $T$, denoted $W_T$, is the maximum number of messages sent by correct processors (combined) between time $T$ and $t_{T}^{\ast}$. 
 
 \item The \emph{worst-case communication complexity} is the worst-case communication complexity after GST$+\Delta$. 
 
 \item The \emph{eventual worst-case communication complexity} is  $ \limsup\limits_{T\rightarrow \infty}  W_T$.  
 
 \item The worst-case latency is the maximum possible value of $t_{\text{GST}}^{\ast}-\text{GST}$. 
 
 \item The eventual worst-case latency is   
  $ \limsup\limits_{T\rightarrow \infty}  \ t_T^*-T$. 
  \end{itemize} 
  We note that Lumiere achieves its eventual worst-case communication complexity and latency for $T$ which is within expected $O(n\Delta)$ time of GST.

\section{Overview of Lumiere} \label{overview} 

In this section, we give an informal overview of the Lumiere protocol. Since the protocol itself is quite simple and hence practical to implement, we start with a brief protocol synopsis in Section \ref{sec:synopsis}. 
 
The insights behind Lumiere and its analysis are more involved. Therefore, following the synopsis, we explain the ingredients that make it work. First, in Section \ref{LP22sec}, we  review the LP22 protocol, and explain why it suffers from the two weaknesses described in Section \ref{intro}. Then, in Section \ref{fev}, we review the basic idea behind the Fever protocol. In Section \ref{bls}, we describe how to combine the techniques developed by LP22 and Fever so as to give a protocol which has $O(n^2)$ worst-case communication complexity and which is smoothly optimistically responsive. Finally, in Section \ref{SS}, we describe how to remove the need for views with $\Omega(n^2)$ communication complexity within time $O(n\Delta)$ of GST. 

\subsection{Lumiere Synopsis} \label{sec:synopsis}

\noindent Borrowing from RareSync and LP22, Lumiere batches views into \emph{epochs}, and intertwines two synchronization procedures: a ``heavy'' epoch synchronization procedure and a ``light'' non-epoch synchronization procedure.

More specifically, at the start of some epochs, Lumiere employs a two round all-to-all broadcast procedure whose quadratic communication cost is amortized over all the views in the epoch. Importantly, Lumiere introduces a new mechanism that prevents performing such epoch synchronizations after a \emph{successful} epoch generating QCs by $2f+1$ leaders. This guarantees that only an expected constant-bounded number of heavy synchronizations will occur after GST. This mechanism is explained in detail in Section \ref{SS}.

Within each epoch, Lumiere employs a light view synchronization procedure, which entails linear message complexity per view, and which allows processors to `bump' their clocks forward and begin the instructions for the next view when they receive a QC -- this process of `bumping' clocks is explained in detail in Sections \ref{fev} and \ref{bls}. Bumping clocks in this way produces a protocol that is smoothly optimistically responsive. 

\vspace{0.1cm}
The above is the entire protocol. However, to make this work we need to tune a parameter $\Gamma$ of the protocol that determines the view timers, so as to guarantee two things.

Firstly, we need that: 

\begin{enumerate}
    \item[(a)] If the $f+1$ honest processors whose clocks are most advanced begin some view with an honest leader within time $\Gamma$ of one another after GST, the leader can generate a QC.  
\end{enumerate}
\noindent Second, because epoch synchronizations stop after a successful epoch, we must also guarantee that, in that event,  either the clocks of honest processors are already sufficiently synchronized,  or else the production of QCs by honest leaders will reduce the gap between the clocks of honest processes so that they are sufficiently synchronized soon after. To this end, we need that: 

\begin{enumerate} 
\item[(b)] The generation of QCs by honest leaders after GST ``shrinks'' the gap between the clocks of the $f+1$ honest processors whose clocks are most advanced (unless this gap is already less than $\Gamma$). 
\end{enumerate} 

It turns out that both (a) and (b) are accomplished simply by tuning the parameter $\Gamma$, and we will discuss appropriate values for $\Gamma$ in the sections that follow.

The protocol is described in detail in Sections \ref{bls}, \ref{SS} and \ref{formal}.

\subsection{Overview of LP22} \label{LP22sec} 
\textbf{Epochs}. The core idea behind LP22 revolves around the concept of epochs: For every $e$,  
the sequence $f+1$ views  $[ e(f+1), \dots, e(f+1)+f]$  is referred to as \emph{epoch} $e$. 
We define $V(e) := e\cdot (f+1)$ and $E(v) := \lfloor v/(f+1) \rfloor$, so that view $V(e)$ is the first view of epoch $e$ and $E(v)$ is the epoch to which view $v$ belongs. The first view of each epoch is also called an \emph{epoch view}, and all other views are called \emph{non-epoch} views. 
The leader of view $v$ is processor $v \text{ mod }n$. 

\vspace{0.2cm} 
\noindent \textbf{The clock time corresponding to a view}. Clock times are not really necessary for specifying the LP22 protocol (which could equally be presented using `timers'), but we wish to give a presentation here which is as similar as possible to our presentation of Lumiere later on. 

As explained in Section \ref{modelsec}, each processor maintains a local clock value $\mathtt{lc}(p)$. 
We also consider a clock time corresponding to each view, denoted $\mathtt{c}_v:=\Gamma v$.  Here $\Gamma$ should be thought of as the length of time allocated to view $v$, and (for LP22) can be set to $\Gamma:=x\Delta +\Delta$ (where $x$ is as specified in Section \ref{modelsec}). Roughly, the idea is that the processor $p$ wishes to execute the instructions for view $v$ once its local clock reaches $\mathtt{c}_v$.

\vspace{0.2cm} 
\noindent \textbf{The instructions for entering epoch views}. Let $v:=V(e)$. Processor $p$ wishes to enter epoch $e$ and view $v$ if it is presently in a lower view and once its local clock reaches $\mathtt{c}_{v}$. At this point, it pauses its local clock and sends an $\mathtt{epoch\ view }\ v$ message to all processors,\footnote{It is convenient throughout to assume that when a processor sends a message to all processors, this includes itself. }  indicating that it wishes to enter epoch $e$.  Upon receiving  $\mathtt{epoch\ view }\ v$ messages from $2f+1$ distinct processors while in a view $<v$, any honest processor combines these into a single threshold signature, which is called an Epoch Certificate (EC) for view $v$, and sends the EC to all processors.  Upon seeing an EC for view $v$ while in any lower view, any honest processor sets $\mathtt{lc}(p):=\mathtt{c}_v$, unpauses its local clock if paused, and  then enters epoch $e$ and view $v$. Note that this process involves honest processors sending $\Theta(n^2)$ messages (combined). If $t\geq $ GST is the first time at which an honest processor enters epoch $e$, then all honest processors see the corresponding EC giving them permission to enter epoch $e$ by time $t+\Delta$.  

\vspace{0.2cm} 
\noindent \textbf{The instructions for entering non-epoch views}. If we did not require a protocol that is optimistically responsive, we could simply have each processor enter non-epoch view $v$ when its local clock reaches $\mathtt{c}_v$. To achieve optimistic responsiveness, LP22 uses a simple trick. Processor $p$ enters non-epoch view $v$ when the first of the the following events occurs: 
\begin{itemize} 
\item Its local clock reaches $\mathtt{c}_v$, or; 
\item Processor $p$ sees a QC for view $v-1$. 
\end{itemize} 

\vspace{0.2cm} 
\noindent \textbf{High level analysis of the protocol}. The key insight is that, while the process for entering epoch $e$ entails honest processors sending $\Theta(n^2)$ messages (combined), no further message sending is then required to achieve synchronization within the epoch. If the first honest processor to enter epoch $e$ does so after GST, then all honest processors will see the corresponding EC within time $\Delta$ of each other. Suppose $v$ is the first view of epoch $e$ with an honest leader. If no leader has already produced a QC for some view $<v$ in the epoch, then view $v$ will produce a QC.

To see that the protocol is optimistically responsive, suppose all leaders of an epoch are honest and let $\delta$ be the actual upper bound on message delay after GST. If  the first honest processor to enter epoch $e$ does so after GST, then (according to the assumptions of Section \ref{modelsec}) the first leader will produce a QC in time $O(\delta)$. All honest processors will receive this QC within time $\delta$ of each other, which means that the leader of the second view will then produce a QC in time $O(\delta)$, and so on. 

It is also not difficult to see that the protocol suffers from the two issues (i) and (ii) outlined in Section \ref{intro}. First of all, we have already noted that entering each epoch requires honest processors to send $\Theta(n^2)$ messages (combined), meaning that the eventual worst-case communication complexity is $\Theta(n^2)$. 
To see that even a single Byzantine processor can cause a latency of $O(n\Delta)$ between consensus decisions, consider what happens when the first $f$ leaders of an epoch produce QCs very quickly. If the last leader of the epoch is Byzantine, then honest processors must then wait for time almost $(f+1)\Gamma$ before wishing to enter the next epoch. Figure~\ref{fig:responsive} illustrates a scenario with three good views producing $QC$s quickly, a faulty fourth view, and the fifth view suffering almost a $3\Gamma$ delay.

\begin{center} 
\begin{figure}[H] 
\begin{center} 
\begin{tikzpicture}[xscale=1,yscale=1]


\draw[-{>[scale=1.5, length=6,
          width=3]},line width=0.4pt] (0,0) -> (10,0); 
          
 \draw[-{>[scale=1.5, length=3,
          width=3]},line width=0.4pt] (0,2.65) -> (1,1.1);          
          
\draw[-{[scale=1.5, length=0.1,
          width=3]},line width=1pt] (1,0) -> (1,0.3);      
          
 
 \draw[-{[scale=1.5, length=0.1,
          width=3]},line width=0.7pt] (3,0) -> (3,0.3);

  \draw[-{[scale=1.5, length=0.1,
          width=3]},line width=0.7pt] (5,0) -> (5,0.3);  
 
  \draw[-{[scale=1.5, length=0.1,
          width=3]},line width=0.7pt] (7,0) -> (7,0.3); 
          
   \draw[-{[scale=1.5, length=0.1,
          width=3]},line width=0.7pt] (9,0) -> (9,0.3);

  
   \draw[-{[scale=1.5, length=0.1,
          width=3]},line width=1.3pt, dashed, orange] (1.7,0) -> (1.7,0.3);

   \draw[-{[scale=1.5, length=0.1,
          width=3]},line width=1.3pt, dashed, orange] (2.5,0) -> (2.5,0.3);

   \draw[-{[scale=1.5, length=0.1,
          width=3]},line width=1.3pt, dashed, orange] (3.2,0) -> (3.2,0.3);

          \node [above] at (0,3) {All-to-all synchronization};
          \node [above] at (0,2.7) {at the beginning of epoch $e$};
          
                 \node [below] at (5,0) {$\mathtt{lc}(p)$};
                 
    \node [above] at (1,0.4) {$V(e)\cdot \Gamma$};
    
       \node [above] at (3,0.4) {$+ \Gamma$};   
       
       \node [above] at (5,0.4) {$+ \Gamma$};   
       
         \node [above] at (7,0.4) {$+ \Gamma$};     
         
            \node [above] at (9,0.4) {$+ \Gamma$};   
            
      \node [above, orange] at (3.5,2) {QCs for views};    
      
        \node [above, orange] at (3.5,1.7) {$V(e),V(e)+1, V(e)+2$};              
                                 
 \draw[-{>[scale=1.5, length=3,
          width=3]},line width=0.4pt, orange] (2.5,1.7) -> (1.7,0.5);     
          
  \draw[-{>[scale=1.5, length=3,
          width=3]},line width=0.4pt, orange] (3.3,1.7) -> (2.5,0.5);           
          
  \draw[-{>[scale=1.5, length=3,
          width=3]},line width=0.4pt, orange] (4,1.7) -> (3.2,0.5);   
          
    \node [above, orange] at (8,3) {Enter view $V(e)+4$};                                                    
                                 
      \node [above, orange] at (8,2.7) {after no progress};      
                                
           \node [above, orange] at (8,2.4) {in view $V(e)+3$};   
           
   \draw[-{>[scale=1.5, length=3,
          width=3]},line width=0.4pt, orange] (8,2.3) -> (9,0.9);

\end{tikzpicture}
\end{center}
\caption{LP22: Epoch-synchronization and optimistically responsive QC generation}
\label{fig:responsive} 
\end{figure}
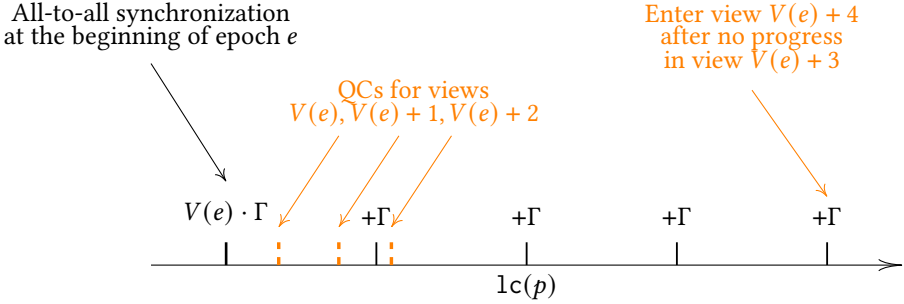
\end{center}

\subsection{Overview of Fever} \label{fev} 
As noted previously, Fever makes stronger assumptions regarding \emph{clock synchronization} than are standard (the assumptions of this paper are standard and are outlined in Section \ref{modelsec}). However, we will show that the techniques developed by Fever can be combined with the LP22 protocol to give a protocol that is smoothly optimistically responsive. 
The fundamental idea behind Fever stems from consideration of what we refer to as the \emph{honest gap}:

\begin{definition} [\textbf{Defining the honest  gaps}] \label{hgdef} At any time $t$, we let $p_{i,t}$ be the honest processor whose local clock is the $i^{\text{th}}$ most advanced, breaking ties arbitrarily. So, $p_{1,t}$ is the honest processor whose local clock is most advanced. For $i\in [1,2f+1]$, we define the $i^{\text{th}}$ \emph{honest  gap} at time $t$ to be $\mathtt{hg}_{i,t}:=  \mathtt{lc}(p_{1,t},t)-\mathtt{lc}(p_{i,t},t)$. In particular, $\mathtt{hg}_{f+1,t}$ is the gap between the local clock of the most advanced honest processor and the local clock of the  $(f+1)^{\text{st}}$ most advanced honest processor. 
\end{definition} 

Recall that, in Section \ref{LP22sec}, $\Gamma$ was the maximum length of time allotted to each view, and that we set $\Gamma:=(x+1)\Delta$ (where $x$ is as specified in Section \ref{modelsec}). For Fever, we set $\Gamma:=2(x+1)\Delta$. We note that this change in $\Gamma$ will not impact protocol performance in the optimistic case that leaders are honest. We also observe below that this factor of 2 can also be decreased to arbitrarily close to 1 by making a simple change to the protocol. 

\vspace{0.2cm} 
\noindent \textbf{The non-standard clock assumption}. The assumption that Fever rests on is that, at the start of the protocol execution, $\mathtt{hg}_{f+1,0}\leq \Gamma$. The protocol also assumes that an honest processor's local clock progresses in real time unless `bumped forward' (according to the protocol instructions). Given this assumption, the protocol  is then designed so that, even though processors often bump their clocks forward: 
\begin{enumerate} 
\item[(a)]  $\mathtt{hg}_{f+1,t}\leq \Gamma$ for all $t\geq 0$. 
\item[(b)]  If $\mathtt{hg}_{f+1,t}\leq \Gamma$ at $t\geq$ GST which is the first time an honest processor enters the view $v$ with honest leader, then the leader will produce a QC. 
\end{enumerate} 
Since the instructions are very simple, we just state them, and then show that they function as intended. 

\vspace{0.2cm}
\noindent  \textbf{Initial and non-initial views}. Fever does not consider any notion of epochs. The leader for view $v$ is processor $\lfloor v/2 \rfloor \text{ mod } n$. If $v$ is even, then $v$ is called `initial', otherwise $v$ is `non-initial'. 
The reason we consider initial and non-initial views will become clear when we come to verify (b) above. 
As in Section \ref{LP22sec}, the clock-time corresponding to view $v$ is $\mathtt{c}_v:=\Gamma v$. 

\vspace{0.1cm} 
\noindent \textbf{When processors enter views}. If $v$ is initial, then $p$ enters view $v$ when  $\mathtt{lc}(p)=\mathtt{c}_v$.  If $v$ is not initial, then $p$ enters view $v$ if it is presently in a view $<v$ and it receives a QC (formed by the underlying protocol) for view $v-1$.

\vspace{0.1cm} 
\noindent  \textbf{View Certificates}. When an honest  processor $p$ enters a view $v$ which is initial, it  sends a $\mathtt{view}\ v$ message to $\mathtt{lead}(v)$. This message is just the value $v$ signed by $p$. Once $\mathtt{lead}(v)$ receives $f+1$ $\mathtt{view}\ v$ messages from distinct processors, it  combines these into a single threshold signature, which is a View Certificate (VC) for view $v$, and sends this VC to all processors.

\vspace{0.1cm} 
\noindent \textbf{When processors bump clocks}.   At any point in the execution, if an honest processor $p$ receives a QC for view $v-1$ (formed by the underlying protocol) or a VC for view $v$, and if $\mathtt{lc}(p)<\mathtt{c}_v$, then $p$ instantaneously bumps their local clock to $\mathtt{c}_v$.

  \vspace{0.3cm} 
 \noindent \textbf{Verifying the claim (a) above}. Since the local clocks of honest processors only ever move forward, it follows that at any point in an execution, if an honest processor $p$ has already contributed to a QC or a VC for view $v$, then $\mathtt{lc}(p)\geq \mathtt{c}_v$. 
To prove the claim, suppose towards a contradiction that  there is a first point of the execution, $t$ say,  at which the claim fails to hold. 
Then it must be the case that some honest $p$ bumps its clock forward at $t$,  and that $p=p_{1,t}$ after bumping its clock forward, with $\mathtt{lc}(p)$ strictly greater than the value of any honest clock at any time $<t$. 

There are two possibilities: 

\begin{enumerate} 
\item $p$ bumps its clock because it receives a VC for some view $v$ with $\mathtt{c}_v>c(p)$.  In this case, there must exist at least one honest processor $p'\neq p$ which contributed to the VC for view $v$. This contradicts the fact that $p=p_{1,t}$. 
\item $p$ bumps its clock because it sees a QC. In this case, at least $f+1$ honest processors must have contributed to the QC, which directly gives the required contradiction.  
\end{enumerate} 

 \vspace{0.2cm} 
 \noindent \textbf{Verifying the claim (b) above}. Suppose that $\mathtt{hg}_{f+1,t}\leq \Gamma$ at $t$ which is the first time an honest processor enters a view $v$ with honest leader. 
 Then $\mathtt{lead}(v)$ will receive $f+1$ $\mathtt{view}\ v$ messages by time $t+\Gamma+\Delta$, and all honest processors will be in view $v$ by time $t+\Gamma +2\Delta$ (if the leader has not already produced a QC for view $v$ by this time): Note that processors do not enter view $v+1$ until seeing a QC for view $v$ because $v+1$ is not initial. All processors will then receive a QC for view $v$ by time $t+2\Gamma-x\Delta$  and will be in view $v+1$ by this time (unless the leader has already produced a QC for view $v+1$). All processors will then receive a QC for view $v+1$ by time $t+2\Gamma$. 

\vspace{0.2cm} 
Although we do not give a formal proof here (we give a formal proof for Lumiere later), it is also clear that the protocol is smoothly optimistically responsive because the delay caused by each faulty leader is at most $\Gamma$ per view. 

\vspace{0.2cm} 
\noindent \textbf{Reducing} $\Gamma$. It is not difficult to see that $\Gamma$ can be made arbitrarily close to $(x+1)\Delta$ by increasing the number of consecutive views allocated to each leader (and altering the definition of initial and non-initial views accordingly): doing so increases worst-case latency, but proportionally decreases the total time that can be wasted by faulty leaders.

\subsection{Basic Lumiere Solution} \label{bls} 

To describe the full version of Lumiere, we break the presentation down into two steps. In this section, we show how to combine the techniques developed by LP22 and Fever to give a protocol called Basic Lumiere, which maintains $O(n^2)$ worst-case communication complexity while also being smoothly optimistically responsive. Then, in Section \ref{SS}, we describe how to modify the protocol to remove the need for views with $\Omega(n^2)$ communication complexity within time $O(n\Delta)$ of GST (this last step turns out to be the one that is complicated). 

\vspace{0.2cm} 
\noindent \textbf{The basic idea}. The idea behind combining LP22 and Fever is simple. Fever requires the assumption that $\mathtt{hg}_{f+1,t}=0$ at the start of the protocol execution. While we do not wish to make this assumption, the `heavy' synchronization process that LP22 employs at the start of each epoch \emph{does} ensure that $\mathtt{hg}_{f+1,t}$ is bounded by $\Gamma$ at the start $t$ of any epoch that begins after GST. We can therefore employ the Fever protocol \emph{within epochs}. Doing so ensures that the $\mathtt{hg}_{f+1,t'}$ remains bounded by $\Gamma$ for $t' \geq t$ within each epoch. All honest leaders therefore produce QCs, and the protocol is smoothly optimistically responsive because each faulty leader can only cause $\Gamma$ delay per view.

\vspace{0.2cm}
\noindent  \textbf{Initial, non-initial, and epoch views}.  The leader for view $v$ is processor $\lfloor v/2 \rfloor \text{ mod } n$. If $v$ is even, then $v$ is called initial: otherwise $v$ is non-initial. If $v \text{ mod } 2(f+1)=0$, then $v$ is called an epoch view. We set $V(e):=2(f+1)e$. The clock-time corresponding to view $v$ is $\mathtt{c}_v:=\Gamma v$. 

\vspace{0.1cm} 
\noindent \textbf{When processors enter initial non-epoch views}. If $v$ is initial and is not an epoch view,  then $p$ enters view $v$ when  $\mathtt{lc}(p)=\mathtt{c}_v$. 

\vspace{0.1cm} 
\noindent \textbf{When processors enter non-initial views}.
If $v$ is not initial, then $p$ enters view $v$ if it is presently in a view $<v$ and it receives a QC  for view $v-1$.

\vspace{0.1cm} 
\noindent \textbf{When processors enter epoch views}. Processor $p$ enters the epoch view $v$ if it is presently in a lower view and if $p$ receives an EC for view $v$. 

\vspace{0.1cm} 
\noindent  \textbf{View Certificates}. When an honest processor $p$ enters a view $v$ which is initial and which is not an epoch view, it  sends a $\mathtt{view}\ v$ message to $\mathtt{lead}(v)$. This message is just the value $v$ signed by $p$. Once $\mathtt{lead}(v)$ receives $f+1$ $\mathtt{view}\ v$ messages from distinct processors, it  combines these into a single threshold signature, which is a view certificate (VC) for view $v$, and sends this VC to all processors.

\vspace{0.1cm} 
\noindent \textbf{Epoch certificates}.  Let $v:=V(e)$. Processor $p$ wishes to enter epoch $e$ and view $v$ if it is presently in a lower view and once its local clock reaches $\mathtt{c}_{v}$. At this point, it pauses its local clock and sends an $\mathtt{epoch\ view }\ v$ message to all processors.  Upon receiving  $\mathtt{epoch\ view }\ v$ messages from $2f+1$ distinct processors while in a view $<v$, any honest processor combines these into a single threshold signature, which is called an EC for view $v$, and sends the EC to all processors. 

\vspace{0.1cm} 
\noindent \textbf{When processors bump clocks}.   At any point in the execution, if a correct processor $p$ receives a QC for view $v-1$ (formed by the underlying protocol) or a VC or EC  for view $v$, and if $\mathtt{lc}(p)<\mathtt{c}_v$, then $p$ instantaneously bumps their local clock to $\mathtt{c}_v$. 
Upon receiving an EC for view $v$ while in a lower view, $p$ unpauses its local clock (if paused).

\subsection{Removing Epoch Synchronisations for the Steady State} \label{SS} 

The idea behind removing the need for repeated `heavy' ($\Theta(n^2)$ communication) epoch view changes after GST is that, once the first has been carried out, honest processors are already sufficiently synchronised. Since GST is unknown, however, processors cannot know directly when synchronization has occurred. Instead, they should look to see when some `success criterion' has been satisfied. For example, one might wait to see an epoch which has produced a QC or a certain number of QCs, and then temporarily pause heavy view changes until one sees an epoch which does not satisfy the success criterion. Unfortunately, taking this approach immediately introduces some complexities: 

\vspace{0.2cm} 
\noindent \textbf{Some processors may see the success criterion satisfied, while others do not}. 
In the case that the success criterion is satisfied due to QCs produced by faulty leaders, it may be the case that some honest processors fail to see the success criterion satisfied. Since they will then require an EC to enter the next epoch, those processors who do see the success criterion satisfied will still need to contribute to the EC. On the other hand, we do not wish Byzantine processors alone to be able to trigger EC formation, otherwise the Byzantine players will be able to cause every epoch to  begin with a heavy view change. 
To deal with this, we have to modify the epoch change process slightly:
\begin{itemize} 
\item  If an honest processor sees the success condition satisfied for an epoch $e$, then they view $V(e+1)$ as a standard initial view (meaning that they enter the view when their local clock reaches $\mathtt{c}_{V(e+1)}$)  and do not immediately send an $\mathtt{epoch\  view}\  V(e+1)$ message (nor pause their local clock).
\item Any honest processor who reaches the end of epoch $e$ and does not see the success criterion satisfied pauses its local clock and sends an  $\mathtt{epoch\  view}\  V(e+1)$ message to all processors 
Any set of $f+1$ $\mathtt{epoch\  view}\  V(e+1)$ messages from distinct processors is referred to as a TC for view $v$. \item When any honest processor in an epoch $\leq e+1$ sees a TC for view $V(e+1)$, they send an $\mathtt{epoch\  view}\  V(e+1)$ message to all processors. 
\item Any processor that does not see the success criterion for epoch $e$ satisfied enters epoch $e+1$ upon seeing an EC for view $V(e+1)$, and unpauses its local clock at that point.
\item An EC for view $V(e+1)$ is now defined to be a set of $\mathtt{epoch\  view}\  V(e+1)$ messages from $2f+1$ distinct processors. 
\end{itemize} 

\vspace{0.2cm} 
\noindent \textbf{The success criterion might be satisfied for all epochs after GST without synchronization actually occurring}. Since QCs may be formed with the help of Byzantine processors, the formation of QCs does not actually imply that the $(f+1)^{\text{st}}$ honest gap is less than $\Gamma$ (or even small). One is therefore potentially presented with a scenario where the success criterion continues to be satisfied for every epoch after GST with the help of Byzantine processors. If epochs involve $2(f+1)$ views, and if the success condition is seeing a single leader in the epoch produce two  QCs, then this would represent highly sub-optimal behaviour. The situation can be improved  by increasing the length of an epoch by a constant factor to $2n$ views (meaning each processor gets two successive views as leader) and by setting the success criterion to be $2f+1$ leaders each producing two QCs (for views in the epoch). If every epoch after GST produces the success criterion, this might now look like a reasonable outcome. Unfortunately, it still represents  sub-optimal behaviour because $f$ honest leaders may fail to produce QCs in each epoch if $f$ Byzantine leaders each produce two QCs (in this case, the adversary is highly over-represented in QC generation).   

To remedy this issue, the basic idea is to set $\Gamma$ so that each honest leader who produces QCs is able to shrink the $(f+1)^{\text{st}}$ honest gap. If the success criterion continues to be satisfied, meaning that multiple honest leaders are shrinking the $(f+1)^{\text{st}}$ honest gap, then the aim is that, within expected time $O(n\Delta)$ of GST, the $(f+1)^{\text{st}}$ honest gap should come down below $\Gamma$.  

\vspace{0.2cm} 
\noindent \textbf{Setting $\Gamma$ to shrink the honest gap}. For technical reasons we set $\Gamma:=2(x+2)\Delta$, but the following argument would also work for the value of $\Gamma$ used in Section \ref{fev}. We insist that honest leaders only produce a QC for view $v$ if they can do it within time $\Gamma/2 - 2\Delta$ of sending the VC for view $v$, or within that time of sending the QC for the previous view if $v$ is not initial.  Note that this bound is on the time at which the QC is produced, rather than the time at which it is received.

 To see that an honest leader who produces a QC after GST shrinks the $(f+1)^{\text{st}}$ honest gap, we can reason as follows: a more formal proof will be given in Section \ref{verification}. Let $t$ be such that $\mathtt{lc}(p_{f+1,t},t)=c_v$ and suppose $p:=\mathtt{lead}(v)$ is honest. The instructions ensure that $p$ receives $f+1$ $\mathtt{view}\ v$ messages by $t+\Delta$ and sends a VC to all processors by this time. The QC is then produced by time $t+\Gamma/2-\Delta$ and is received by all processors by time $t+\Gamma/2$. Upon receiving this QC, honest processors whose local clocks are less than $\mathtt{c}_{v+1}$ forward their clocks to this value. This reduces the $(f+1)^{\text{st}}$ honest gap (and also the $(2f+1)^{\text{st}}$) by at least $\Gamma/2$ or to a value below $\Gamma$ unless,  for some $t'\in [t,t+\Gamma/2]$, some honest processor bumps its honest clock forward upon seeing a QC at $t'$ and becomes $p_{1,t'}$ upon doing so. In the latter case the $(f+1)^{\text{st}}$ honest gap is anyway reduced to below $\Gamma$. 
 
 \vspace{0.2cm} 
\noindent \textbf{The honest gap must be brought below $\Gamma$ within a single epoch}. While the $(f+1)^{\text{st}}$ honest gap cannot be increased (except to a value below $\Gamma$) within an epoch, the same is not true of the $(2f+1)^{\text{st}}$ honest gap.  Unfortunately, the fact that some honest processors may see epoch $e$ produce the success criterion, while others do not, means that the process of moving to epoch $e+1$ may significantly increase the $(f+1)^{\text{st}}$ honest gap: If only $f$ honest processors see the success criterion, then some honest processors will have to wait to see an EC for view $V(e+1)$ before entering the next epoch. This means that the $(f+1)^{\text{st}}$ honest gap can increase to become equal to the $(2f+1)^{\text{st}}$ honest gap. 

This difficulty can be addressed by increasing the length of an epoch by a constant factor, so that a single epoch that produces the success criterion brings the $(f+1)^{\text{st}}$ honest gap to within $\Gamma$ by the end of the epoch. We note that increasing the length of the epoch does represent a slight tradeoff. While Basic Lumiere outperforms LP22 on all measures, the cost for Lumiere of significantly decreasing the eventual worst-case communication complexity is a constant factor increase in the worst-case latency, i.e.\ the time to the \emph{first} consensus decision after GST. Latency in the steady state is still significantly decreased when compared to LP22, because Lumiere is smoothly optimistically responsive.  

 \vspace{0.2cm} 
\noindent \textbf{Two further complexities}. Extending the length of an epoch allows us to ensure that the $(f+1)^{\text{st}}$ honest gap is brought down to below $\Gamma$ by the last view of the epoch. 
If the $(2f+1)^{\text{st}}$ honest gap is still large at that point, however, there remains the danger that the process of changing epoch will substantially increase the $(f+1)^{\text{st}}$ honest gap: We wish to ensure that the $(f+1)^{\text{st}}$ honest gap remains small thereafter, so that all honest leaders from QCs from that point on. To ensure that the epoch change does not increase the  $(f+1)^{\text{st}}$ honest gap (or at least not by more than $\Delta$), it suffices that the last leader of the epoch and the first leader of the next are both honest, since then they will be able to reduce the  $(2f+1)^{\text{st}}$ honest gap to below $\Gamma$. To make things simple, we can also ensure that these two views have the same leader. 

To describe a final complexity that needs to be dealt with, let us suppose that the $(f+1)^{\text{st}}$ honest gap is less than $\Gamma$ at the beginning of an epoch $e$, so that all honest leaders produce QCs during the epoch. If time less than $\Delta$ passes after the success criterion is satisfied and before the clocks of some honest processors reach $\mathtt{c}_{V(e+1)}$ (because a long sequence of QCs are produced in time less than $\Delta$),  then it is possible that honest processors will not have seen the QCs produced by some honest leaders when their local clock reaches $\mathtt{c}_{V(e+1)}$, and will therefore send $\mathtt{epoch\ view}\ V(e+1)$ messages. To remedy this obstacle, we have processors wait time $\Delta$ before sending an $\mathtt{epoch\ view}\ V(e+1)$ message when they do not immediately see satisfaction of the success criterion.

\section{The formal specification} \label{formal}

It is convenient to assume that, whenever any processor sends a message to all processors, it also sends this message to itself (and this message is immediately received). 

\vspace{0.2cm} 
\noindent \textbf{Local clocks.}
Recall that each processor $p$ has a \emph{local clock} value, denoted $\mathtt{lc}(p)$. Initially, $p$ sets $\mathtt{lc}(p) := 0$. For simplicity, we suppose $\mathtt{lc}(p)$ advances in real time (with zero drift after GST), except when $p$ pauses $\mathtt{lc}(p)$ or bumps it forward. As noted in Section \ref{modelsec}, our analysis is easily modified to deal with a scenario where local clocks have bounded drift during any interval after GST in which they are not paused or bumped forward. 
Recall also Definition \ref{hgdef}, which applies unaltered to this section.

\vspace{0.2cm} 
\noindent \textbf{Leaders and the clock time corresponding to each view}. To determine the leader of view $v$, let $(g_0,\dots, g_{z-1})$ be a random ordering of the permutations of $\{ 1,\dots,n \}$ subject to the condition that, if $i\leq z$ is odd, then $g_{i}$ and $g_{i+1 \text{ mod }z}$ are reverse orderings.\footnote{The latter condition on reverse orderings is stipulated so that, for $e\geq 0$, the last leader of epoch $e$ is the same as the first leader of epoch $e+1$. } We wish to give each leader two consecutive views, and to order leaders according to $g_0$, then $g_1$, and so on, cycling back to $g_0$ once we have ordered according to $g_{z-1}$. To achieve this, set $j:= \lfloor v/(2n) \rfloor  \text{ mod } z$.  The leader for view $v$, denoted $\mathtt{lead}(v)$,  is processor $g_{j}(\lfloor v/2 \rfloor \text{ mod } n)$.

\vspace{0.2cm} 
\noindent \textbf{Initial and non-initial views}.  If $v$ is even, then view $v$ is called \emph{initial}. Otherwise, $v$ is \emph{non-initial}. The clock time corresponding to view $v$ is $\mathtt{c}_v:=\Gamma v$. We set $\Gamma:=2(x+2)\Delta$, and insist that honest leaders only produce a QC for view $v$ if they can do it within time $\Gamma/2 - 2\Delta$ of sending the VC for view $v$, or within that time of sending the QC for the previous view if $v$ is not initial.
The  operational distinction between initial and non-initial views is expanded on below.  

\vspace{0.2cm} 
\noindent \textbf{Epochs and epoch views}.  Epoch $e$ consists of the $10n$ views $[5(2n)e,5(2n)(e+1))$. 
The first view of each epoch is called an \emph{epoch view}, while other views are referred to as \emph{non-epoch} views. We define $V(e):=10ne$ and $E(v):= \lfloor v/(10n) \rfloor$, so that $V(e)$ is the first view of epoch $e$ and $E(v)$ is the epoch to which view $v$ belongs. 


\vspace{0.2cm} 
\noindent \textbf{When processors enter initial views}. A processor $p$ enters the initial view $v$ when $\mathtt{lc}(p)==\mathtt{c}_v$ if its local clock is not paused. Upon doing so, $p$ sends a $\mathtt{view}\ v$ message to $\mathtt{lead}(v)$. 

\vspace{0.2cm} 
\noindent \textbf{Forming VCs}. Suppose $v$ is an initial (epoch or non-epoch) view. 
If $\mathtt{lead}(v)$ is presently in view $v'\leq v$ and receives  $\mathtt{view}\ v$ messages from $f+1$ distinct processors, then $\mathtt{lead}(v)$  forms a threshold signature which is a VC for view $v$ and sends this to all processors.

\vspace{0.2cm} 
\noindent \textbf{The instructions upon receiving a VC}. 
If $p$ is presently in view $v'< v$ and receives a VC for initial view $v$, then $p$ sets $\mathtt{lc}(p):=\mathtt{c}_v$.

\vspace{0.2cm} 
\noindent \textbf{When processors enter non-initial views}. A processor $p$ enters the non-initial view $v$ if it is presently in a lower view and if it sees a QC for view $v-1$. 

\vspace{0.2cm} 
\noindent \textbf{The instructions upon receiving a QC}. If $p$ is presently in view $v$ and sees a QC for view $v'\geq v$, then it sets $\mathtt{lc}(p):=\mathtt{c}_{v'+1}$ if $\mathtt{lc}(p)< \mathtt{c}_{v'+1}$. 

\vspace{0.2cm} 
The  operational distinction between initial and non-initial views is this.  Upon their local clocks reaching $\mathtt{c}_v$, where $v$ is an initial non-epoch view, processors perform light view synchronization to bring others into view $v$. Namely, they enter view $v$ and immediately send a view message. On the other hand, processors do not enter the non-initial view $v+1$ unless they obtain a QC for view $v$; view $v+1$ serves as sort of a ``grace period'' allowing the initial view $v$ preceding it to produce a QC even after some local clocks have reached $\mathtt{c}_{v+1}$. Note that, upon obtaining a QC for view $v$, processors bump their local clocks (if lower) to $\mathtt{c}_{v+1}$ and proceed to enter view $v+1$ immediately if in a lower view. If a processor receives a QC for the non-initial view $v+1$ while in a view $\leq v+1$, then it immediately enters view $v+2$ (unless, perhaps, $v+2$ is an epoch view -- see below). 

 \vspace{0.2cm} 
\noindent \textbf{The success criterion}. Each processor $p$ maintains a local variable $\mathtt{success}(e)$, initially 0. Processor $p$ sets $\mathtt{success}(e):=1$ upon seeing at least $2f+1$  distinct processors each produce 10 QCs for views in the epoch. We say that epoch $e$ `produces the success criterion' if at least $2f+1$  distinct processors each produce 10 QCs for views in the epoch.  

 \vspace{0.2cm} 
\noindent \textbf{ECs and TCs}. Suppose $v=V(e)$.  An EC for view $v$ is a set of $2f+1$ $\mathtt{epoch\  view}\ v$ messages, signed by distinct processors. A TC for view $v$ is a set of $f+1$ $\mathtt{epoch\  view}\ v$ messages, signed by distinct processors.

 \vspace{0.2cm} 
\noindent \textbf{The instructions for entering epoch views}. Suppose $v=V(e)$. When $\mathtt{lc}(p)=\mathtt{c}_v$, there are two cases: 
\begin{itemize}
\item  Upon first seeing $\mathtt{lc}(p)=\mathtt{c}_v$ and $\mathtt{success}(E(v)-1) = 0$: 
\begin{itemize} 
\item $p$ pauses its local-clock until seeing an EC, QC or VC for a view $\geq v$, or a TC for a view $>v$, or until $\mathtt{success}(E(v)-1) =1$ ;
\item If its local clock is still paused time $\Delta$ after pausing, then $p$ sends  an $\mathtt{epoch\ view}\  v$ message to all processors. 
\end{itemize}

\item Upon first seeing $\mathtt{lc}(p)=\mathtt{c}_v$ and $\mathtt{success}(E(v)-1) = 1$ (and also under other conditions already stated in the bullet point above), $p$ enters epoch $e$ and view $V(e)$. 
\end{itemize}

 \vspace{0.2cm} 
\noindent \textbf{Forming ECs}. Suppose $v=V(e)$. As stipulated above, $p$ sends an  $\mathtt{epoch\ view}\ v$ message to all processors $\Delta$ time after its local clock reaches  $\mathtt{lc}(p)=\mathtt{c}_v$ if $\mathtt{success}(e)=0$ at this time. We further stipulate that if $p$ is in epoch $e'\leq e$ and receives a TC for view $v$, then $p$ sets $\mathtt{lc}(p)=\mathtt{c}_v$ if $\mathtt{lc}(p)<\mathtt{c}_v$ and sends an $\mathtt{epoch\ view}\ v$ message to all processors.

\vspace{0.4cm} 
The pseudocode is shown in Algorithm \ref{alg:lum}.

\vspace{0.4cm}
The proof of correctness (establishing Theorem \ref{maintheorem}) appears in Section \ref{verification}. 

\begin{algorithm*} 
\begin{algorithmic}[1]
\scriptsize
    \State \textbf{Local variables}  \label{lum:setup}
    \State $\mathtt{lc}(p)$, initially 0 \Comment{This is the value of $p$'s local-clock }
    \State $view(p)$, initially -1   \Comment{The present view of $p$.} 
    \State $epoch(p)$, initially -1  \Comment{The present epoch of $p$.} 
    \State $\mathtt{success}(e)$, $e\in \mathbb{Z}_{\geq -1}$, initially 0 \Comment{Updated as described above}\label{lum:setup2}
    \State

    \State \Comment{--------- Epoch Synchronization Instructions (start) ----------------------------------}\label{lum:EC}
    \State 
   \State  \textbf{Upon} first seeing $\mathtt{lc}(p)==\mathtt{c}_v$ for epoch view $v>view(p)$  and $\mathtt{success}(E(v)-1) == 0$: 
    \label{lum:pause}
    	    \State \hspace{0.5cm} Pause local-clock $\mathtt{lc}(p)$ until seeing an EC, QC or VC for a view $\geq v$ or a TC for a view $>v$, or until $\mathtt{success}(E(v)-1) ==1$ ; \label{pauseline} 
	    \State \hspace{0.5cm} If local clock is still paused time $\Delta$ after pausing, send  an $\mathtt{epoch\ view}\  v$ message to all processors;

	    \State 
 
    \State  \textbf{Upon} first seeing $\mathtt{lc}(p)==\mathtt{c}_v$ for epoch view $v>view(p)$  and $\mathtt{success}(E(v)-1) == 1$: \label{condy1} 
    
    	    \State \hspace{0.5cm} Set $epoch(p):=E(v)$ and $view(p):=v$; \label{viewalt0}

 \State 

       \State  \textbf{Upon} first seeing a TC for epoch view $v$ with $E(v)\geq epoch(p)$:  \label{seef+1}
    
          \State \hspace{0.5cm}  \textbf{If} $\mathtt{lc}(p) < \mathtt{c}_v$ \textbf{then}:
           \State \hspace{1cm}  For each initial view $v'$ with $view(p)\leq v' <v$ send a $\mathtt{view}\ v'$ message to $\mathtt{lead}(v')$ if not already sent;
            \State \hspace{1cm}    Set $\mathtt{lc}(p) := \mathtt{c}_v$; \label{bump1} 
    
    \State \hspace{0.5cm} \textbf{If} $view(p)<v-1$ \textbf{then} set  $view(p):=v-1$ and $epoch(p):=E(v)-1$;  \label{viewalt1}

        \State \hspace{0.5cm}Send  an $\mathtt{epoch\ view}\  v$  message to all processors if not already sent;

\State

      \State  \textbf{Upon} first seeing an EC for epoch view $v$ with $E(v) > epoch(p)$: 
    \label{lum:unpause}
    

         \State \hspace{0.5cm}  Set $view(p):=v$ and $epoch(p):=E(v)$; \label{viewalt2}

            \State \Comment{--------- Epoch Synchronization Instructions (end) -----------------}\label{lum:EC2}

    \State \Comment{--------- View Synchronization Instructions (start) -----------------------------------}\label{lum:VC}
    
    \State 
    
    \State \textbf{Upon} $\mathtt{lc}(p)==\mathtt{c}_v$ for $v$ initial and $epoch(p)==E(v)$:
    
         \State \hspace{0.5cm} \textbf{If} $view(p)<v$, set $view(p):=v$; \label{viewalt3}  
    
    \State \hspace{0.5cm} Send a $\mathtt{view}\ v$ message to $\mathtt{lead}(v)$;

     \State 
     
     \State \textbf{If} $p==\mathtt{lead}(v)$ for initial view $v\geq view(p)$: 
     
     \State \hspace{0.5cm} \textbf{Upon} first seeing $\mathtt{view}\ v$ messages from $f+1$ distinct processors:

      \State \hspace{1cm} Form a VC for view $v$ and send to all processors; 
      
                    
      \State 
      
      \State \textbf{Upon} first seeing a VC for initial view $v>view(p)$; 
      
    \State \hspace{0.5cm}  \textbf{If} $\mathtt{lc}(p) < \mathtt{c}_v$ \textbf{then}:
       
           \State \hspace{1cm}  For each initial view $v'$ with $view(p)\leq v' <v$ send a $\mathtt{view}\ v'$ message to $\mathtt{lead}(v')$ if not already sent;

        \State \hspace{1cm} Set   $\mathtt{lc}(p) := \mathtt{c}_v$;  \label{bump2} 
    
    \State \hspace{0.5cm} Set  $view(p):=v$, $epoch(p):=E(v)$; \label{viewalt4}

     \State \Comment{--------- View Synchronization Instructions (end) ------------------}

     \State \Comment{--------- bump local-clock forward on QC -------------------------------------------------}
     
     \State 
     
     \State \textbf{Upon} first seeing a QC for view $v\geq view(p)$: 
     
     \State \hspace{0.5cm} \textbf{If} $\mathtt{lc}(p)<\mathtt{c}_{v+1}$ \textbf{then}:
     
                \State \hspace{1cm}  For each initial view $v'$ with $view(p)\leq v' <v$ send a $\mathtt{view}\ v'$ message to $\mathtt{lead}(v')$ if not already sent;

        \State \hspace{1cm}  Set $\mathtt{lc}(p):= \mathtt{c}_{v+1}$.  \label{bump3} 
     
          \State \hspace{0.5cm} \textbf{If} $v+1$ is a non-epoch view \textbf{then} set $view(p):=v+1$ and $epoch(p):=E(v+1)$;  \label{viewalt5} 
          
                    \State \hspace{0.5cm} \textbf{If} $v+1$ is an epoch view and $view(v)<v$ \textbf{then} set $view(p):=v$ and $epoch(p):=E(v)$;  \label{viewalt6}

\end{algorithmic}
    \caption{The instructions for processor $p$}
    \label{alg:lum}
\end{algorithm*} 

\section{Proof of correctness} \label{verification} 

We say that processor $p$  is in  epoch $e$ at time $t$ if the  local variable $epoch(p)=e$ at time $t$.
We say that $p$  \emph{enters} epoch $e$ at time $t$ if  $t$ is the least time at which $p$ is in epoch $e$. We talk about views similarly.  

We begin with a sequence of three simple lemmas that formally establish conditions on the relationship between $t$, $view(p)$, $epoch(p)$ and $\mathtt{lc}(p)$. 

\begin{lemma} \label{view/epoch} 
If honest $p$ is in view $v$ and epoch $e$ at time $t$, then $E(v)=e$. 
\end{lemma} 
\begin{proof} 
 All honest processors begin in view -1 and epoch -1. Since $E(-1)=-1$, the hypothesis therefore holds at initialization.  Now consider any instruction carried out at time $t\geq 0$. The view or epoch that $p$ is in can only be altered by lines \ref{viewalt0}, \ref{viewalt1}, \ref{viewalt2}, \ref{viewalt3}, \ref{viewalt4}, \ref{viewalt5}, or \ref{viewalt6}. In each case, the instruction explicitly sets the view $v$ and epoch $e$ of $p$ so that $E(v)=e$. 
\end{proof} 

The following lemma is required to show that (1) from the definition of the view synchronization task in Section \ref{modelsec} is satisfied. 

\begin{lemma} 
If honest $p$ is in view $v_1$ and epoch $e_1$ at time $t_1$, and is in  view $v_2$ and epoch $e_2$ at time $t_2>t_1$, then $e_2\geq e_1$,  $v_2\geq v_1$, and $\mathtt{lc}(p,t_2)\geq \mathtt{lc}(p,t_1)$. 
\end{lemma} 
\begin{proof} 
Suppose  honest $p$ is in view $v_1$ and epoch $e_1$ at time $t_1$ and is in  view $v_2$ and epoch $e_2$ at time $t_2>t_1$. That $\mathtt{lc}(p,t_2)\geq \mathtt{lc}(p,t_1)$ follows from the fact that $\mathtt{lc}(p)$ advances in real time, unless paused (line \ref{pauseline}) or bumped strictly forward (lines \ref{bump1}, \ref{bump2}, or \ref{bump3}). 

To show that $e_2\geq e_1$ and $v_2\geq v_1$, consider the various instructions that can alter $view(p)$ or $epoch(p)$. 

\vspace{0.1cm} 
\noindent Lines \ref{viewalt0}, \ref{viewalt1}, \ref{viewalt4}, \ref{viewalt5}, \ref{viewalt6}.   If one of these instructions  sets $view(p):=v$ and $epoch(p):=E(v)$ for some $v$,  then the stated condition for carrying out the instruction explicitly requires that, before this update, $p$ is in a  view $<v$. By Lemma \ref{view/epoch}, $p$ is also in an epoch $\leq E(v)$ before the update. 


 \vspace{0.1cm} 
\noindent Line \ref{viewalt2}. If this instruction sets $view(p):=v$ and $epoch(p):=E(v)$, then the stated condition for carrying out the instruction requires that,  before this update, $p$ is in an epoch $<E(v)$.  By Lemma \ref{view/epoch}, $p$ is also in a view $<v$ before the update. 

 \vspace{0.1cm} 
\noindent Line \ref{viewalt3}. This instruction does not change $epoch(p)$. If the instruction sets $view(p):=v$,  then the stated condition for carrying out the instruction explicitly requires that, before this update, $p$ is in a view $<v$. 
%
%
\end{proof} 

\begin{lemma} \label{viewint} 
Consider $v_0\geq 0$ which is initial. If $p$ is honest then: 
\begin{enumerate} 
\item[(i)] If $p$ is in view $v_0$ then $\mathtt{lc}(p)\in [\mathtt{c}_{v_0}, \mathtt{c}_{v_0+2}]$. 
\item[(ii)] If $p$ is in view $v_0+1$ then $\mathtt{lc}(p)\in [\mathtt{c}_{v_0+1}, \mathtt{c}_{v_0+2}]$.
\end{enumerate} 
\end{lemma} 
\begin{proof} 
The proof is by induction on (discrete) time and is trivially satisfied at initialization. First, consider the instructions that can redefine $view(p)$. 

\vspace{0.1cm} 
\noindent Lines \ref{viewalt0} and \ref{viewalt3}. If either of these instructions sets $view(p):=v$ for some $v$, then $\mathtt{lc}(p)=\mathtt{c}_v$. 

\vspace{0.1cm} 
\noindent Lines \ref{viewalt1} and \ref{viewalt6}. If either of these instruction sets $view(p):=v$ for some $v$, then before this update $p$ is in a view $<v$. It follows from the induction hypothesis that $\mathtt{lc}(p)\leq \mathtt{c}_{v+1}$. Lines \ref{bump1} and \ref{bump3} therefore ensure that $\mathtt{lc}(p) = \mathtt{c}_{v+1}$. 

\vspace{0.1cm} 
\noindent Line \ref{viewalt2}. If this instruction sets $view(p):=v$, then $epoch(p)<E(v)$ prior to the update. By Lemma \ref{view/epoch} and the induction hypothesis it follows that $\mathtt{lc}(p)\leq \mathtt{c}_v$ before the update. 
Since the condition of line \ref{seef+1} must be already be satisfied for epoch view $v$ when $p$ sees an EC for view $v$, line \ref{bump1} ensures that $\mathtt{lc}(p)=\mathtt{c}_v$ in this case.

\vspace{0.1cm} 
\noindent Lines \ref{viewalt4}. If this instruction sets $view(p):=v$, then before this update $p$ is in a view $<v$. It follows from the induction hypothesis, and since $v$ is initial, that $\mathtt{lc}(p)\leq \mathtt{c}_v$. Line \ref{bump2} therefore ensures that $\mathtt{lc}(p)=\mathtt{c}_v$ in this case. 

\vspace{0.1cm} 
\noindent Line \ref{viewalt5}. If this instruction sets $view(p):=v+1$, then $p$ was in view $\leq v$ before the update. It follows from the induction hypothesis that $\mathtt{lc}(p)\leq \mathtt{c}_{v+2}$.  
Line \ref{bump3} therefore ensures that $\mathtt{lc}(p)\in [\mathtt{c}_{v+1},\mathtt{c}_{v+2}]$ in this case.

\vspace{0.2cm} 
Now let $v_0$ and $p$ be as in the statement of the lemma. The analysis above shows that if any instruction sets $view(p):=v_0$, then $\mathtt{lc}(p)\in [\mathtt{c}_{v_0},\mathtt{c}_{v_0+2}]$ at that point, and that if any instruction sets  $view(p):=v_0+1$, then $\mathtt{lc}(p)\in [\mathtt{c}_{v_0+1},\mathtt{c}_{v_0+2}]$ at that point. To complete the proof of the lemma, it suffices to show that if there exists a first time $t_1$ at which $\mathtt{lc}(p)>\mathtt{c}_{v_0+2}$, the instructions set $view(p)>v_0$ by that time. So, suppose there exists such a time $t_1$. There are then two cases to consider: 

\vspace{0.1cm} 
\noindent \textbf{Case 1}. If $p$ receives an EC, QC or VC for any view  $\geq v_0+2$, or a TC for any view $>v_0+2$,   at any time $\leq t_1$ while in view $v$, then $view(p)$ is set to a value $\geq v_0+2$ at that point (by one of lines \ref{viewalt1}, \ref{viewalt2},   \ref{viewalt4}, \ref{viewalt5}, or \ref{viewalt6}).

\vspace{0.1cm} 
\noindent \textbf{Case 2}. Case 1 does not hold. Then, since $p$ only bumps its clock past $\mathtt{c}_{v_0+2}$ upon seeing an EC, QC, VC, or TC of the kind described in Case 1,  there must exist a least $t_0<t_1$ at which $\mathtt{lc}(p)=\mathtt{c}_{v_0+2}$. If $v_0+2$ is a non-epoch view, then line \ref{viewalt3} sets $view(p):=v_0+2$ at $t_0$. If view $v_0+2$ is an epoch view $V(e+1)$, then consider the possible subcases. 
If the local variable $\mathtt{success}(e)=0$ at $t_0$, then $p$'s local clock is paused at $t_0$. If the local clock becomes unpaused because  $p$ sees either a QC, a VC or an EC for a view $\geq v_0+2$, or a TC for a view $>v_0+2$, then this contradicts the existence of $t_1$ and the fact that Case 1 does not hold.  
It must therefore hold that either $\mathtt{success}(e)=1$ at $t_0$ or that $p$'s local clock is paused at $t_0$ and becomes unpaused upon seeing $\mathtt{success}(e)=1$. Line \ref{viewalt0} then sets $view(p):=v_0+2$ at $t_0$. 
\end{proof}

Next, we prove a sequence of lemmas which bound the time honest processors spend in each epoch after GST.

\begin{lemma} \label{enterbefore} 
If $p$ is honest and enters epoch $e$, then at least $f+1$ honest processors must previously have entered epoch $e-1$. 
\end{lemma} 
\begin{proof} 
We prove the lemma when $p$ is the first honest processor to enter epoch $e$, then it a fortiori holds for honest processors entering the epoch later. If $p$ enters epoch $e$ upon seeing an EC for view $V(e)$, then at least $f+1$ honest processors must have contributed to that EC while in epoch $e-1$. If $p$ enters epoch $e$ upon seeing $\mathtt{success}(e-1)=1$, then at least one view in epoch $e-1$ must have already produced a QC, and at least $f+1$ honest processors must have contributed to that QC while in epoch $e-1$. 
\end{proof}

\begin{lemma} \label{fin0} 
If an honest processor is in epoch $e$ at $t\geq$ GST, then all honest processors are in epochs $\geq e-1$ by time $t+\Delta$. 
\end{lemma} 
\begin{proof} 
We prove the lemma when $p$ is the first honest processor to enter epoch $e$, then it a fortiori holds for honest processors entering the epoch later. If $p$ enters epoch $e$ upon seeing an EC for view $V(e)$, then all honest processors see a TC for view $V(e)$ by $t+\Delta$. If $p$ enters epoch $e$ upon seeing the success criterion satisfied for epoch $e-1$, then some honest leader must have produced a QC for a view in epoch $e-1$, which will be seen by all honest processors by time $t+\Delta$. 
\end{proof}

\begin{lemma} \label{catchup} 
If an honest processor is in epoch $e$ at $t\geq$ GST, then all honest processors are in epochs $\geq e$ by time $t+O(n\Delta)$. 
\end{lemma} 
\begin{proof} 

 
Let $p_1$ be the first honest processor to enter epoch $e$.  Then $p_1$ enters epoch $e$ either because they see an EC for view $V(e)$, or because they see the success criterion satisfied for epoch $e-1$. In the first case, all honest processors see a TC for view $V(e)$ by time $t+\Delta$. If any honest processor has entered an epoch   $>e$ by this time, then the claim of the lemma follows from Lemma \ref{fin0}. So, suppose otherwise. Then all honest processors  send an $\mathtt{epoch\ view}\ V(e)$ message upon seeing the TC (if they have not already done so). All honest processors therefore see an EC for view $V(e)$ and enter an epoch $\geq e$ by time  $t+2\Delta$.  In the second case, some honest leader must have produced a QC for one of the last $2n$ views in epoch $e-1$ (actually for one of the last $4f+1$ views of the epoch, but the bound $2n$ suffices for our calculations below), which means that all honest processors see that QC and are in an epoch $\geq e-1$ by time  $t+\Delta$ (and are within the last $2n-1$ views of the epoch if in epoch $e-1$).

There are now two further sub-cases to consider, within the second case above: 
 
\vspace{0.1cm} 
\noindent Case 1. At least $f+1$ of the honest processors wait for a heavy epoch synchronization to enter epoch $e$.  
Then those processors produce a TC for view $V(e)$ by $t+2\Delta +(2n-1)\Gamma$. If any honest processor has entered an epoch $>e$ by this time, then the claim of the lemma follows from Lemma \ref{fin0}. So, suppose otherwise.  Then all honest processors have sent an $\mathtt{epoch\ view}\ V(e)$ message by time $t+3\Delta +(2n-1)\Gamma$. All honest processors therefore see an EC for view $V(e)$ by time $t+2n\Gamma$. 
 
\vspace{0.1cm} 
\noindent Case 2. At least $f+1$ honest processors treat view $V(e)$ as a standard initial view (meaning that they enter the view when their local clock reaches $\mathtt{c}_{V(e)}$).
If any honest processor enters epoch $e+1$ by time $t+12n\Gamma$, then the claim follows from Lemma \ref{fin0}. 
Otherwise, all honest processors see a TC for view $V(e+1)$ and enter epoch $e$ by time $t+12n\Gamma$.
 \end{proof}

\begin{lemma} \label{infepochs}
If an honest processor is in epoch $e$ at $t\geq$ GST, then all honest processors are in epochs $\geq e+1$ within time $t+O(n\Delta)$. 
\end{lemma} 
\begin{proof} 
By Lemma \ref{catchup}, all honest processors are in an epoch $\geq e$ by time $t+O(n\Delta)$, and it suffices to show that some honest processor enters epoch $e+1$ by time $t+O(n\Delta)$. So suppose otherwise. Then all honest processors send an $\mathtt{epoch\ view}\ V(e+1)$ message within time $t+O(n\Delta)$, which gives the required contradiction. 
\end{proof} 

\noindent 
\noindent \textbf{Specifying when epochs begin and end}. The clock time $\mathtt{c}_v$ is used to determine when $p$ should enter view $v$, according to its local clock. However, we also wish to specify a \emph{global} time corresponding to each view: $\mathtt{vt}_{v}$ is defined to be the least time at which at least $f+1$ honest processors are in views $\geq v$. Lemma \ref{infepochs}  shows that $\mathtt{vt}_{v}$ is defined for all $v$. Epoch $e$ begins at time $\mathtt{start}_e:=\mathtt{vt}_{V(e)}$ and ends at time $\mathtt{end}_e$, which is defined to be the first time $t$ at which $\mathtt{lc}(p_{f+1,t})\geq \mathtt{c}_{V(e+1)}$.

From the definitions and by Lemma \ref{viewint}, it is immediate that $\mathtt{start}_{e+1}\geq \mathtt{end}_e\geq \mathtt{start}_{e}$ for all $e\geq -1$. Lemma \ref{enterbefore} shows that  $\mathtt{start}_{e+1}>\mathtt{start}_e$. 

\vspace{0.2cm} 
\noindent \textbf{Further terminology: timely starts}. Suppose $v$ is initial and set $t:=\mathtt{vt}_{v}$.  We say that $v$ has a \emph{timely start} if the first honest processor to enter any view $\geq v$ does so after GST and   if $\mathtt{hg}_{f+1,t}\leq \Gamma +2\Delta$.  We say that epoch $e\geq 0$ has a timely start if view $V(e)$ has a timely start. 

\begin{lemma} \label{produce} 
Suppose $v\geq 0$ is initial and $p:=\mathtt{lead}(v)$ is honest. If view $v$ has a timely start then, for $t:=\mathtt{vt}_{v}$: 
\begin{itemize} 
\item All honest processors receive a QC for view $v$ by time $t+\Gamma/2$. 
\item All honest processors receive a QC for view $v+1$ by time $t+\Gamma-2\Delta$. 
\end{itemize} 
\end{lemma} 
\begin{proof} 
Note that no honest processor enters view $v+1$ before seeing a QC for view $v$. 
Since $p$ must receive $\mathtt{view}\ v$ messages from at least $f+1$ distinct processors by time $t+\Delta$, all honest processors see a VC for view $v$ by time $t+2\Delta$. Honest processors will therefore be in view $v$ by this time, unless a QC has already been produced for view $v$. By our choice of $\Gamma$, it follows that all honest processors will receive a QC for view $v$ by time $t+\Gamma/2$, as claimed, and will enter view $v+1$ by this time. By the choice of $\Gamma$, all honest processors will then see a QC for view $v+1$ by time $t+\Gamma-2\Delta$, as claimed. 
\end{proof} 

\noindent \textbf{Further terminology: primary clock bumps}.  If honest $p$ bumps its clock forward at $t$ and if $p=p_{1,t}$ after doing so, with $\mathtt{lc}(p)$ strictly greater than the value of any honest clock at any time $<t$, then we say that a \emph{primary clock bump} occurs at $t$, and that $p$ implements a primary clock bump at $t$.

\begin{lemma} \label{nodecrease} The two following statements hold: 
\begin{enumerate} 
\item If a primary clock bump occurs at $t$, then $\mathtt{hg}_{f+1,t}\leq \Gamma$. 
\item The honest  gap $\mathtt{hg}_{f+1,t}$ does not increase during epoch $e$, except to a value below $\Gamma$, i.e.\ if $t,t'\in [\mathtt{start}_e,\mathtt{end}_e]$ and $t<t'$, then either $\mathtt{hg}_{f+1,t}\geq \mathtt{hg}_{f+1,t'}$, or else $\mathtt{hg}_{f+1,t'}\leq \Gamma$. 
\end{enumerate} 
\end{lemma} 
\begin{proof} 
Statement (2) follows from (1) and the fact that there is no time  $t\in [\mathtt{start}_e,\mathtt{end}_e)$ at which the local clock of processor $p_{f+1,t}$ is paused. 

To establish (1), suppose that $p$ implements a primary clock bump at $t$.  It cannot be that $p$ bumps its clock upon seeing a VC for some view $v$ (in an epoch $\geq e$), because the existence of the VC implies that some honest processor $p'\neq p$ already contributed to that VC, meaning that  $\mathtt{lc}(p)$ is not strictly greater than the value of any honest clock at any time $<t$. By the same argument, it cannot be that $p$ bumps its clock upon seeing a TC. So, $p$ must bump its local clock because it sees a QC for some view $v$ at $t$. Then at least $f+1$ honest processors must have contributed to that QC. By Lemma \ref{viewint}, each of those processors $p'$ must satisfy $\mathtt{lc}(p',t)\geq \mathtt{c}_v$, meaning that $\mathtt{hg}_{f+1,t}\leq \Gamma$, as required. 
\end{proof}

\begin{lemma} \label{2Delta}
For $e\geq 0$, suppose the  first honest processor to enter epoch $e$ does so at time $t\geq$ GST upon seeing an EC for view $v:=V(e)$. Then epoch $e$ has a timely start.
\end{lemma} 
\begin{proof} 
 Since $t\geq $GST and at least $f+1$ honest processors must have contributed to the EC,  all honest processors see a TC for view $v$ by time $t+\Delta$. 
Upon seeing a TC for view $v$ by time $t+\Delta$, any honest processor will send an  $\mathtt{epoch\  view}\ v$ message, unless they have already done so, or unless they are already in an epoch $>e$ by time $t+\Delta$. If any honest processor is in an epoch $>e$ by time $t+\Delta$, then all views in epoch $e$ must produce QCs. In particular, a QC for view $v$ must be produced before time $t+\Delta$, meaning that at least $f+1$ honest processors enter view $v$ before this QC is produced and view $v$ has a timely start (in fact $\mathtt{hg}_{f+1,\mathtt{start}_e}\leq  \Delta$). If no honest processor is in an epoch $>e$ by time $t+\Delta$, then all honest processors send $\mathtt{epoch\  view}\ v$ messages by that time, and all honest processors see an EC for view $v$ by time $t+2\Delta$. Again, we conclude that view $v$ has a timely start (in fact  $\mathtt{hg}_{f+1,\mathtt{start}_e}\leq  2\Delta$ in that case). 
\end{proof}

\begin{lemma} \label{largebound} 
Suppose the  first honest processor to enter epoch $e\geq 0$ does so after GST. Then $\mathtt{hg}_{f+1,\mathtt{start}_e}< (4f+2)\Gamma$.  
\end{lemma} 
\begin{proof} 
Suppose first that the first honest processor to enter epoch $e$ does so upon seeing an EC for view $v:=V(e)$. Then the claim follows from Lemma \ref{2Delta}. 

If the first case does not apply, then the first processor $p$ to enter epoch $e$ must do so  at some time $t$, say, upon seeing $\mathtt{success}(e-1)=1$. Note that: 

\begin{enumerate} 
\item[$(\dagger)$:] If there exists a least $t'>t$ with $t'\leq \mathtt{end}_e$ at which a primary clock bump occurs, then some view in epoch $e$ must already have produced a QC, meaning that $\mathtt{start}_e<t'$ and $\mathtt{hg}_{f+1,\mathtt{start}_e}<t'-t$. 
\end{enumerate}. 

 From the fact that $p$ sees epoch $e-1$ satisfy the success criterion at $t$, it follows that there exists some view in the range  $[V(e)-2(2f+1),V(e))$ which has honest leader and which has produced a QC by time $t$. Since all honest processors see this QC by time $t+\Delta$, all honest processors have local clock values at least $\mathtt{c}_{V(e)}-(4f+1)\Gamma$ by $t+\Delta$, and have clock values at least $\mathtt{c}_{V(e)}$ by time $t_1:=t+ \Delta + (4f+1)\Gamma$. 
 
 If a primary clock bump occurs in the interval $(t,t+(4f+2)\Gamma]$, then the lemma follows from $(\dagger)$ above. So suppose otherwise, and consider two cases: 
 
 \vspace{0.1cm} 
 \noindent \textbf{Case 1}: At least $f+1$ honest processors send $\mathtt{epoch \ view}\  V(e)$ messages by time $t_1+\Delta$. In this case, all honest processors send  $\mathtt{epoch \ view}\  V(e)$ 
 messages by time $t_1+2\Delta$ and enter epoch $e$ by time $t_1+3\Delta$. Since $\Gamma >3\Delta$,  the claim of the lemma follows in that case. 
 
  \vspace{0.1cm} 
 \noindent \textbf{Case 2}: Otherwise. Then at least $f+1$ honest processors must see epoch $e-1$ satisfy the success criterion and enter epoch $e$ by time $t_1+\Delta$. The claim of the lemma therefore follows similarly in that case. 
\end{proof} 

\begin{lemma} \label{shorten} 
Suppose the  first honest processor to enter epoch $e\geq 0$ does so after GST.   Suppose further that initial view $v$ is amongst the first $8n$ views of epoch $e$ and that $p:=\mathtt{lead}(v)$ is honest. If  QCs are produced for views $v$ and $v+1$ then either: 
\begin{enumerate} 
\item[(i)] $\mathtt{hg}_{f+1,\mathtt{vt}_{v'}}\leq \text{max}\{ \Gamma, \mathtt{hg}_{f+1,\mathtt{vt}_{v}}-\Gamma  \} $ for all initial views $v'>v$ in epoch $e$, or; 
\item[(ii)] The last initial view $v^*$ of the epoch has a timely start and $\mathtt{hg}_{f+1,\mathtt{vt}_{v^*}}\leq \Gamma$. 
\end{enumerate} 
\end{lemma} 
\begin{proof} 
By Lemmas \ref{largebound} and \ref{nodecrease}, we know that   $\mathtt{hg}_{f+1,t}< (4f+2)\Gamma$ for all $t\in [\mathtt{start}_e,\mathtt{end}_e]$. Set $t:=\mathtt{vt}_v$. Since view $v$ produces a QC, it follows from $(\diamond_2)$ in Section \ref{modelsec} that $\mathtt{lc}(p_{f+1,t},t)=\mathtt{c}_v$, otherwise there does not exist any non-empty time interval during which at least $2f+1$ processors (some of them potentially Byzantine) can be in view $v$ and act honestly. 
Since view $v$ produces a QC,  $p$ must produce a VC for view $v$ by time $t+\Delta$: since $p$ receives $f+1$ $\mathtt{view}\ v$ messages by $t+\Delta$, it produces a VC by this time unless it is already in a greater view, but if $p$ enters a view greater than $v$ before producing a QC, then no QC can be produced for view $v$. 

There are now two cases to consider: 

\vspace{0.1cm} 
\noindent \textbf{Case 1}. A primary bump occurs at some first time $t'\in (t,t+\Gamma]$. In this case, it follows from Lemma \ref{nodecrease} that $\mathtt{hg}_{f+1,t'}\leq \Gamma$. Since  $\mathtt{hg}_{f+1,t}< (4f+2)\Gamma$ and $v$ is amongst the first $8n$ views of epoch $e$, it also follows from Lemma \ref{nodecrease} that the last initial view $v^*$ of the epoch has a timely start and $\mathtt{hg}_{f+1,\mathtt{vt}_{v^*}}\leq \Gamma$.

\vspace{0.1cm} 
\noindent \textbf{Case 2}. Case 1 does not occur. In this case, since $p$ produces QCs for views $v$ and $v+1$, all honest processors receive a QC for view $v+1$ by time $t+\Gamma$. At this point, all honest processors have local clock values at least $\mathtt{c}_{v+2}$.  If $v'\geq v+2$ and $\mathtt{vt}_{v'}\leq t+\Gamma$, it follows directly that  $\mathtt{hg}_{f+1,\mathtt{vt}_{v'}}\leq  \mathtt{hg}_{f+1,\mathtt{vt}_{v}}-\Gamma $. 
If  $v'\geq v+2$ and $\mathtt{vt}_{v'}\geq t+\Gamma$ then it follows from Lemma \ref{nodecrease} that $\mathtt{hg}_{f+1,\mathtt{vt}_{v'}}\leq \text{max}\{ \Gamma, \mathtt{hg}_{f+1,\mathtt{vt}_{v}}-\Gamma  \} $, as required. 
\end{proof} 

\begin{lemma} \label{lastleader} 
Suppose the  first honest processor to enter epoch $e\geq 0$ does so after GST.  If the last leader of epoch $e$ is honest, then epoch $e+1$ has a timely start. 
\end{lemma} 
\begin{proof} 
If epoch $e$ does not produce the success criterion, then the claim follows from Lemma \ref{2Delta}. 
So, suppose otherwise. By Lemma \ref{largebound}, we have that $\mathtt{hg}_{f+1,\mathtt{vt}_{V(e)}}< (4f+2)\Gamma$. Since epoch $e$ produces the success criterion, there exist at least $4(f+1)$ initial views $v$  amongst the first $8n$ views of epoch $e$ which have honest leader and such that QCs are produced for views $v$ and $v+1$. From Lemma \ref{shorten} (either because clause (ii) applies at least once, or because clause (i) applies in every instance), it follows that the last initial view of epoch $e$ has a timely start. 

Let $v$ be the last initial view of epoch $e$. By Lemma \ref{produce}, $p:=\mathtt{lead}(v)$ produces QCs for views $v$ and $v+1$. Let $t$ be the time at which $p$ sends a QC for view $v+1$ to all processors and note that no honest processor can enter epoch $e+1$ prior to $t$. All honest processors receive a QC for view $v+1$ by time $t+\Delta$. If at least $f+1$ correct processors enter epoch $e+1$ upon seeing the success criterion satisfied, then $p$ (who is also the leader of view $V(e+1)$) produces a VC for view $V(e+1)$ by time $t+3\Delta$ and all correct processors enter epoch $e+1$ by time $t+4\Delta$. In this case, epoch $e+1$ has a timely start. Otherwise, at least $f+1$ correct processors must send $\mathtt{epoch\ view}\  V(e+1)$ messages by time $t+2\Delta$. All correct processors see a TC for view $V(e+1)$ by time $t+3\Delta$. If $p$ has already produced a QC for view $V(e+1)$ by time $t+4\Delta$ then $\mathtt{start}_{e+1}<t+4\Delta$ and epoch $e+1$ has a timely start. Otherwise all honest processors send an  $\mathtt{epoch\ view}\  V(e+1)$ message by time $t+3\Delta$, all honest processors see an EC for view $V(e+1)$ by time $t+4\Delta$, and epoch $e+1$ has a timely start. 
\end{proof}

\begin{lemma} \label{finlem1} 
Let $e$ be least such that the first honest processor to enter epoch $e$ does so after GST.  Let $e_1\geq e$ be the least such that the last leader of epoch $e_1$ is honest. 
\begin{enumerate} 
\item If every epoch in $[e-1,e_1]$ produces the success criterion, then  epoch $e_1+1$ has a timely start.
\item If there exists a least $e_2\in [e-1,e_1]$ such that epoch $e_2$ does not produce the success criterion, then epoch $e_2+1$ has a timely start. 
\end{enumerate} 
\end{lemma} 
\begin{proof}
Claim (1) follows from Lemma \ref{lastleader}. Claim (2) follows from Lemma \ref{2Delta}. 
\end{proof} 

\begin{lemma} \label{finlem2} 
For $e\geq 0$, suppose epoch $e$ has a timely start.  Then: 
\begin{enumerate} 
\item Every view with honest leader in epoch $e$ produces a QC. 
\item No honest processor sends an $\mathtt{epoch\ view}\ V(e+1)$ message. 
\item Epoch $e+1$ has a timely start. 
\end{enumerate} 
\end{lemma} 
\begin{proof} 
Claim (1) follows from Lemmas \ref{nodecrease} and \ref{produce}. From Lemmas \ref{shorten} and  \ref{nodecrease}, it then follows that $\mathtt{hg}_{f+1,\mathtt{end}_e}\leq \Gamma$ and that every honest processor sees the success criterion for epoch $e$ satisfied within time $\Delta$ of its local clock reaching $\mathtt{c}_{V(e+1)}$. It follows that no correct processor sends an   $\mathtt{epoch\ view}\ V(e+1)$ message and that $\mathtt{hg}_{f+1,\mathtt{start}_{e+1}}\leq \Gamma +\Delta$. So, claims (2) and (3) also hold. 
\end{proof}

\begin{theorem} Lumiere has worst-case latency $O(n\Delta)$, worst case communication complexity $O(n^2)$, eventual worst-case communication complexity $O(nf_a+n)$ and eventual worst-case latency $O(f_a\Delta +\delta)$.
\end{theorem} 
\begin{proof} 
 The result follows from Lemmas \ref{fin0}, \ref{infepochs}, \ref{finlem1} and \ref{finlem2}, since honest processors send $O(n^2)$ messages (combined) in any individual epoch. 
\end{proof}

\section{Related Work}
\label{sec:related-work}

HotStuff~\cite{yin2019hotstuff} was the first BFT SMR  protocol to separate the view synchronization module and the core consensus logic.
HotStuff named the view synchronization module the ``PaceMaker'' and left its implementation unspecified.
While HotStuff requires 3 round trips within each view, HotStuff-2~\cite{malkhi2023hotstuff} reduces this number to 2 round trips.

Cogsworth~\cite{Cogsworth21} was the first to formalize BVS as a separate problem, and provided an algorithm with expected $O(n)$ communication complexity and expected $O(1)$ latency in runs with benign failures, but with sub-optimal performance in the worst case. Naor and Keidar~\cite{NK20} improved Cogsworth to runs with Byzantine faults and, when combined with Hotstuff, produced the first BFT SMR protocol with expected linear message complexity in partial synchrony. Both protocols suffer from sub-optimal cubic complexity in the worst case.  

Several papers by Bravo, Chockler, and Gotsman~\cite{bravo2020making2,bravo2022liveness,bravo2022making} define a framework for analyzing the liveness properties of SMR protocols. 
These papers do not attempt to optimize performance, but rather introduce a general framework in the partial synchrony model to allow better comparison of protocols. For example, they describe PBFT~\cite{castro1999practical} within this framework.

Two recent protocols, RareSync~\cite{opodis22} and LP22~\cite{L22}, both solve BVS with optimal $O(n^2)$ communication complexity in the worst case (providing cryptographic assumptions hold), thereby finally matching the lower bound established by Dolev and Reischuk in 1985 \cite{dolev1985bounds}. 
LP22 also achieves optimistic responsiveness. However, both of these protocols suffer from two major issues. First, neither protocol is smoothly optimistically responsive. RareSync is not optimistically responsive. While LP22 is optimistically responsive, even a single Byzantine processor may infinitely often cause $\Omega(n\Delta)$ latency between consecutive consensus decisions. Second, even in the absence of Byzantine action, infinitely many views require honest processors to send $\Omega(n^2)$ messages. Albeit this communication overhead being amortized across $O(n)$ decisions, it may cause periodic slowdowns. Ideally, one would hope for worst-case complexity between every pair of consensus decisions which is (after some finite time after GST)  $O(f_an+n)$. For LP22 and RareSync, the fact that  infinitely many views require honest  processors to send $\Omega(n^2)$ messages means that the corresponding bound is $O(n^2)$.

Fever \cite{lewis2023fever} is another recent protocol, which operates in a different model than RareSync and LP22. While Fever makes standard assumptions regarding message delivery in the partial synchrony model, the protocol requires stronger than standard assumptions on \emph{clock synchronization}.  Specifically, Fever assumes that there is a known bound on the gap between the clocks of honest  processors at the start of the protocol execution, and that the clocks of honest processors suffer bounded drift prior to GST. Under these stronger and non-standard assumptions, Fever achieves optimal $O(n^2)$ communication complexity in the worst case and also addresses the two issues raised above. In comparison, Lumeiere achieves these results under standard clock assumptions.

\section{Final comments} 

Lumiere introduces two fundamental innovations. The first of these combines techniques from \cite{lewis2023fever} and \cite{L22,opodis22} to give a protocol with $O(n^2)$ worst-case communication complexity and which has eventual worst-case latency $O(f_a\Delta +\delta)$. 

The second innovation removes the need for repeated heavy epoch changes, and results in a protocol 
with eventual worst-case communication complexity $O(nf_a+n)$. Since implementing this second change requires a constant factor increase in epoch length, it is most practically useful in contexts where periods of asynchrony are expected to be occasional and where synchrony reflects the standard network state. 

We leave it as an open question as to whether it is possible to achieve a protocol with the same  worst case communication complexity, eventual worst-case communication complexity  and eventual worst-case latency as Lumiere, but which also achieves better than $O(n\Delta)$ worst-case latency. 

%

\bibliographystyle{plainurl}

\end{document}